\RequirePackage{ifpdf}
\ifpdf 
\documentclass[pdftex]{sigma}
\else
\documentclass{sigma}
\fi

\begin{document}

\allowdisplaybreaks

\renewcommand{\thefootnote}{$\star$}

\renewcommand{\PaperNumber}{046}

\FirstPageHeading

\ShortArticleName{Noncommutative Nonlinear Schr\"odinger Equation}

\ArticleName{The Scattering Problem for a Noncommutative \\ Nonlinear
Schr\"odinger Equation\footnote{This paper is a
contribution to the Special Issue ``Noncommutative Spaces and Fields''. The
full collection is available at
\href{http://www.emis.de/journals/SIGMA/noncommutative.html}{http://www.emis.de/journals/SIGMA/noncommutative.html}}}

\Author{Bergfinnur DURHUUS~$^\dag$ and Victor GAYRAL~$^\ddag$}

\AuthorNameForHeading{B.~Durhuus and V.~Gayral}

\Address{$^\dag$~Department of Mathematics, Copenhagen University,\\
\hphantom{$^\dag$}{}~Universitetsparken 5,
DK-2100 Copenhagen {\O}, Denmark}
\EmailD{\href{mailto:durhuus@math.ku.dk}{durhuus@math.ku.dk}}

\Address{$^\ddag$~Laboratoire de Math\'ematiques, Universit\'e de Reims Champagne-Ardenne,\\
\hphantom{$^\ddag$}{}~Moulin de la Housse - BP 1039 51687 Reims cedex 2, France}
\EmailD{\href{mailto:victor.gayral@univ-reims.fr}{victor.gayral@univ-reims.fr}}

\ArticleDates{Received March 03, 2010, in f\/inal form May 20, 2010;  Published online June 03, 2010}

\Abstract{We investigate scattering properties of a Moyal deformed version of the
nonlinear Schr{\"o}dinger equation in an even number of space dimensions.
With rather weak conditions on the degree of nonlinearity, the
Cauchy problem for general initial data has a unique globally def\/ined
solution, and also has solitary wave solutions if the interaction  potential
is suitably chosen.
We demonstrate how to set up a scattering framework for equations of
this type, including appropriate decay estimates of the free time evolution
and the construction of wave operators def\/ined for small scattering
data in the general case and for arbitrary scattering data in the
rotationally symmetric case.}


\Classification{35K99; 58B34; 53D55}

\section{Introduction}

The scattering problem for general nonlinear
f\/ield equations has been intensively studied for many years with
considerable progress, a seminal result being the establishment of
asymptotic completeness in energy space for the
nonlinear Schr{\"o}dinger equation in space dimension \mbox{$n\geq 3$} with
interaction $|\varphi|^{p-1}\varphi$, where
$1+4/n\leq p\leq 1+4/(n-2)$, and analogous results for the nonlinear Klein--Gordon
equation~\cite{GiVe}. Still, many important questions remain
unanswered, in particular relating to the existence of wave operators
def\/ined on appropriate subspaces of scattering states and to
asymptotic completeness for more general interactions, although many
partial results have been obtained. We refer to~\cite{Strauss,Ginibre} for an overview.
One source of dif\/f\/iculties encountered can be traced back to the
singular nature of pointwise
multiplication of functions with respect to appropriate Lebesgue or
Sobolev norms. It therefore appears natural to look for f\/ield equations where
this part of the problem can be eliminated while preserving as much
as possible of the remaining structure. One such possibility
is to deform the standard product appropriately, and  our
purpose  in the present article is to exploit this idea.

Deformations of the kind mentioned applied to relativistic f\/ield
equations have turned out to be of relevance for understanding
non-perturbative aspects of string theory, see e.g.~\cite{GMS,HKL,DJN}.
We shall limit ourselves to studying a deformed version of the
nonlinear Schr{\"o}dinger equation in an even number $n=2d$ of space
dimensions. The rather crude Hilbert space
techniques~\cite{Reed,RS} we apply allow us to consider polynomial
interactions only. We hope that an application of more modern techniques
\cite{Caz} will allow a treatment of a larger class of interactions.
Our main purpose will be to demonstrate that, for
the deformed equation,  the Cauchy problem is in a certain sense more
regular as compared to the
classical equation as well as to set up a natural scattering framework,
 including appropriate decay estimates of the free time evolution
and construction of wave operators under rather mild conditions on the
interaction.

The deformed, or
noncommutative, version of the nonlinear Schr{\"o}dinger
equation (NCNLS) in $2d$ space dimensions can be written as
\begin{gather}
\label{NLSch}
\big(i\partial_t-{\Delta}\big)\varphi(x,t) =
\varphi\,\star_\theta F_{\star_\theta}(\varphi^*\star_\theta\varphi)(x,t) ,
\end{gather}
where $\star_\theta$ denotes the Moyal product (see below) of functions of
$2d$ space variables $x$, $F_{\star_\theta}$ denotes a real polynomial
with respect to this product and ${\Delta}=-\sum_{i=1}^{2d}\partial_i^2$ is
the standard Laplacian in~$2d$ variables.
The Moyal product considered here is given by
\begin{gather*}
f\star_\theta g(x):=(2\pi)^{-2d}\int e^{-iyz} f\left(x-\frac\Theta2
y\right)  g(x+z)\,d^{2d} y\,d^{2d} z .
\end{gather*}
The constant skew-symmetric $(2d\times2d)$-matrix $\Theta$ is assumed to
be given in the canonical form
\[
\Theta  =  \theta \begin{pmatrix} 0 & I_d\\ -I_d & 0\end{pmatrix} ,
\]
where $I_d$ denotes the $d\times d$ identity matrix and
$\theta>0$ is called the deformation parameter.

Def\/ining
$\varphi_\theta(x,t) = \varphi(\theta^\frac 12 x,\theta t)$,
we have the scaling identity
\[
(\varphi\star_\theta \psi)_{\theta} = \varphi_{\theta}\star_1 \psi_{\theta}.
\]
It follows that $\varphi$ satisf\/ies (\ref{NLSch}) if and only if
\begin{gather*}
\big(i\partial_t-\Delta\big)\varphi_\theta(x,t) =
\theta \varphi_\theta \star F_\star(\varphi_\theta^*\star\varphi_\theta)(x,t) ,
\end{gather*}
where we have dropped the subscript on the $\star$-product when $\theta =1$.

The $\star$-product is intimately connected to the so-called Weyl
quantization map $W$. This map associates an operator
$W(f)$ on $L^2(\mathbb R^d)$ to appropriate functions (or distributions) $f$ on $\mathbb R^{2d}$
such that the kernel $K_W(f)$ of $W(f)$ is given by
\begin{gather*}
K_W(f)(x,y)=(2\pi)^{-d}\int_{\mathbb
R^d}f\left(\frac{x+y}{2},p\right)e^{i(x-y)\cdot p}\,dp  .
\end{gather*}
 It is easily seen that $W$ is an
isomorphism from $L^2(\mathbb R^{2d})$ onto the Hilbert space $\cal H$ of
Hilbert--Schmidt operators on $L^2(\mathbb R^d)$ fulf\/illing
\begin{gather}
\label{Wunitary}
\Vert W(f)\Vert_2 = (2\pi)^{-d/2}\Vert f\Vert_{L^2(\mathbb R^{2d})} ,
\end{gather}
where $\Vert\cdot\Vert_2$ denotes the Hilbert--Schmidt norm. Moreover,
$W$ maps the space ${\cal S}(\mathbb R^{2d})$ onto the space of
operators whose kernel is a Schwartz function and its relation to the
$\star$-product is exhibited by the identity
\[
W(f\star g) = W(f)W(g) .
\]
Further properties of the $\star$-product can be found in
e.g.~\cite{GGBISV}.

By use of $W$, equation~\eqref{NLSch} can be restated (see also \cite{DJN})
as a dif\/ferential equation
\begin{gather}
\label{NLSch2}
i\partial_t\phi - 2 \sum_{k=1}^d[a_k^*,[a_k,\phi]] =\theta \phi F(\phi^*\phi) ,
\end{gather}
 for  the operator-valued function $\phi(t) := W(\varphi_\theta(\cdot,t))$,
where we have introduced the creation and annihilation operators
\[
a_k =\frac{1}{\sqrt 2}(x_k + \partial_k)\qquad\mbox{and}\qquad a_k^* =
\frac{1}{\sqrt 2}(x_k - \partial_k) .
\]
The operator
\[
{\bf \Delta} = 2 \sum_{k=1}^d \mbox{ad}\, a_k^*\;\mbox{ad}\, a_k,
\]
with domain $D({\bf\Delta})$,
is def\/ined in a natural way as a self-adjoint operator on $\cal H$ by the
relations
\begin{gather}
\label{tria}
{\bf \Delta} =  W\Delta W^{-1} ,\qquad D({\bf \Delta})=WD(\Delta) ,
\end{gather}
where $\Delta$ denotes the standard self-adjoint $2d$-dimensional Laplace
operator with maximal domain $D(\Delta)=H_2^2(\mathbb R^{2d})$.

We shall primarily be interested in globally def\/ined \emph{mild}
solutions to the Cauchy problem associated to
equation (\ref{NLSch2}), that is continuous solutions $\phi:\mathbb
R\to{\cal H}$ to the corresponding integral equation
\begin{gather}
\label{NLSch3}
\phi(t) = e^{-i(t-t_0){\bf \Delta}}\phi_0 - i\int_{t_0}^t e^{i(s-t){\bf \Delta}}
\phi(s) F\big(|\phi(s)|^2\big)\,ds .
\end{gather}
This equation is weaker than (\ref{NLSch2}) in the sense that if
$\phi: I\to D({\bf \Delta})$ is a continuously dif\/ferentiable solution to
(\ref{NLSch2}) def\/ined on an interval $I$ containing $t_0$, i.e.\ a
\emph{strong solution} on $I$, then it also fulf\/ills (\ref{NLSch3}).
This latter equation  naturally f\/its into
the standard Hilbert space framework for evolution equations, see
e.g.~\cite{RS,Reed}. The following theorem is proven in Section~\ref{secth1}
below.

\begin{theorem}\label{th1} Let $F$ be an arbitrary polynomial over
  $\mathbb R$.
\begin{itemize}\itemsep=0pt
\item[$a)$]  For every $\phi_0\in{\cal H}$, the
equation \eqref{NLSch3} has a unique continuous
solution $\phi:\mathbb R\to {\cal H}$.  For every $\phi_0\in D({\bf\Delta})$, the
equation \eqref{NLSch2} has a unique strong
solution $\phi:\mathbb R\to D({\bf\Delta})$ such that $\phi(t_0)=\phi_0$.
\item[$b)$]  Assume $F$ is a polynomial over $\mathbb R$ with positive
highest order coefficient that has a unique local minimum $x_0$ on the
positive  real line and that $F(x_0)< F(0)$.
If $\theta$ is large enough, there exist oscillating solutions to
\eqref{NLSch2} of the form
\[
\phi(t) = e^{i\omega(t_0-t)}\phi_0 ,\qquad t\in\mathbb R ,
\]
for suitably chosen initial data $\phi_0\in D({\bf\Delta})$ and frequency
$\omega\in\mathbb R$.
\end{itemize}
\end{theorem}

This result follows by a slight adaptation of the methods of \cite{DJN,DJ}.
 The essential new aspect of $a)$, as compared to the corresponding
result for the  classical case \cite[Theorem 3.2]{Strauss}, is its validity for
interaction polynomials without
restrictions on the degree of nonlinearity. Solitary waves as in b)
 are well known to exist for the classical nonlinear
Schr{\"o}dinger equation with attractive interaction
$-|\varphi|^{p-1}\varphi$ in the range $1<p<1+4/(n-2)$ for  $n\geq 3$ \cite{CaLi,GSS,ShStrauss}.

 In order to formulate our main results on the scattering problem we
need to introduce ap\-propriate Hilbert spaces and auxiliary norms.
By $\vert n\rangle$, $n=(n_1,\dots,n_k)\in\mathbb N_0^d$ we shall denote the
standard orthonormal basis for $L^2(\mathbb R^d)$ consisting of
eigenstates for the $d$-dimensional harmonic oscillator, and by
$(\phi_{mn})$ we denote the matrix representing a bounded operator
$\phi$ with respect to this basis, that is
\[
\phi = \sum_{m,n}\phi_{mn}\vert n\rangle\langle m|,\qquad \phi_{mn}:=\langle n|\phi|m\rangle .
\]
Let
\[
b_{mn} := 1+|m-n| ,\qquad m,n\in \mathbb N_0^d ,
\]
where $|\cdot|$ denotes the Euclidean norm in $\mathbb R^d$. For
exponents $p\geq 1$, $\alpha\in\mathbb R$ and $\phi$ as above, we def\/ine
the norms $\Vert\cdot \Vert_{p,\alpha}$  by
\begin{gather}
\label{alphanorms}
\Vert\phi\Vert_{p,\alpha}^p := \sum_{m,n}b_{mn}^{\alpha p}|\phi_{mn}|^p ,
\quad p<\infty, \qquad\mbox{and}\qquad
\Vert\phi\Vert_{\infty,\alpha}:=\sup_{m,n}\{b_{mn}^\alpha|\phi_{mn}|\} ,
\end{gather}
and we let ${\cal L}_{p,\alpha}$ denote the space of operators
on $L^2(\mathbb R^d)$ for which
$\Vert\cdot \Vert_{p,\alpha}$ is f\/inite. We note that
$\mathcal H_{\alpha}:=\mathcal L_{2,\alpha}$ is a  Hilbert space with
scalar product
\[
\langle\phi,\psi\rangle_\alpha:=\sum_{m,n}b_{mn}^{2\alpha} \overline{\phi_{mn}}
 \psi_{mn} ,
\]
and, in particular, ${\cal H}_{0}$ equals  the space $\cal H$ of
Hilbert--Schmidt operators on $L^2(\mathbb R^d)$ with $\Vert\cdot\Vert_{2,0} =
\Vert\cdot\Vert_2$.
Clearly, the operator $U_\alpha:\mathcal H_\alpha\to \mathcal H$ def\/ined by
\begin{gather}\label{Ualpha}
(U_\alpha\phi)_{mn}=
b^{\alpha}_{mn}\phi_{mn} ,\qquad \phi\in\mathcal H_\alpha ,
\end{gather}
is unitary, and the operator
${\bf\Delta}_\alpha := U_\alpha^*\, {\bf \Delta}\,U_\alpha$, with domain
$D({\bf\Delta}_\alpha):=U_\alpha^* D({\bf \Delta})$,
 is self-adjoint on $\mathcal H_\alpha$.

We use the notation $\Vert\cdot\Vert_{\rm op}$ for
the standard operator norm and def\/ine the norm $\Vert\cdot\Vert_a$ by
\[
\Vert\phi\Vert_a := \Vert\widetilde\phi\Vert_{\rm op} ,
\]
whenever the operator
\[
 \widetilde\phi := \sum_{m,n}|\phi_{mn}|\, |n\rangle\langle m|
\]
 is bounded.

With this notations we have the following decay estimate for the free
propagation:
\begin{theorem}\label{th2} For $\alpha\geq 0$ and  $\beta > d$ there exists a constant
$c_\beta>0$ such that
\begin{gather}\label{freedecay}
\Vert\, e^{-it{\bf\Delta}_\alpha}\phi \Vert_a  \leq  c_{\beta}(1+
|t|)^{-\frac{d}{2}} \Vert\phi\Vert_{1,\beta} ,
\qquad \phi\in {\cal H}_\alpha\cap {\cal L}_{1,\beta} .
\end{gather}
\end{theorem}

This estimate should be compared to the standard one used for the
classical nonlinear wave equations with the $L^\infty$-norm on the
left and a $L^1$-Sobolev norm on the right, applied to
functions of $2d$ space variables. In this case, one obtains rather
trivially a decay exponent~$d$ instead of~$d/2$. However, those norms
are not well behaved with respect to the $\star$-product and cannot be
used for our purposes. Instead, we have to work a bit harder to establish
(\ref{freedecay}) using uniform estimates on the classical Jacobi
polynomials. This is accomplished in Section~\ref{secth2}.

Our last main result concerns the existence of wave operators
def\/ined on a scattering subspace $\Sigma_\alpha\subset \mathcal H_\alpha$. In order to def\/ine
the latter we introduce, for f\/ixed $\alpha\geq 0$, the scattering norm
of $\phi\in{\cal H}_\alpha$ by
\begin{gather*}
|||\phi|||_{\alpha}:=
\Vert\phi\Vert_{2,\alpha}+\sup_{t\in\mathbb R} |t|^{\frac d2}
\Vert e^{-it{\bf\Delta}_\alpha}\phi\Vert_a
\end{gather*}
and set
\begin{gather*}
\Sigma_{\alpha}:=\left\{\phi\in\mathcal H_\alpha\,|\;|||\phi|||_{\alpha}<\infty\right\}\,.
\end{gather*}
From \eqref{alphanorms}, we immediately see that
$\mathcal L_{p,\alpha}\subset\mathcal L_{q,\beta}$, if $p\leq q$ and
$\alpha p\geq\beta q$, so that Theorem \ref{th2} gives, in particular, the
inclusion $\mathcal L_{1,\beta}\subset\Sigma_\alpha$ if $\beta> \max\{d, 2\alpha\}$.

The relevant integral equations to solve in a scattering situation
correspond to initial/f\/inal data $\phi_\pm$ at $t=\pm\infty$ and
take the form
\begin{gather}
\label{NLSch4-}
\phi(t)=e^{-it{{\bf\Delta}_\alpha}}\phi_--i\int_{-\infty}^t
e^{i{(s-t){\bf\Delta}_\alpha}}\phi(s) F\big(|\phi(s)|^2\big)\,ds ,
\end{gather}
and
\begin{gather}
\label{NLSch4+}
\phi(t)=e^{-it{{\bf\Delta}_\alpha}}\phi_++i\int^{+\infty}_t
e^{i{(s-t){\bf\Delta}_\alpha}}\phi(s) F\big(|\phi(s)|^2\big)\,ds ,
\end{gather}
respectively, where $F$ has been redef\/ined to include $\theta$.
 {\it We assume here for simplicity that the polynomial $F$
has no constant term}, since such a term could trivially be incorporated by
subtracting it from ${\bf\Delta}_\alpha$ in the preceding formulas.   The
f\/irst problem
then is to determine spaces of initial/f\/inal data $\phi_\pm$, such that the
equations above have a unique global solution and such that these
solutions behave as free solutions for $t\to \pm\infty$.
To this end we establish the following in Section~\ref{secth3}.

\begin{theorem}\label{th3}
Let $\alpha>2d$ and assume that the polynomial $F$ has no constant
term and, in addition, no linear term
if $d=1$ or $d=2$.

Then there exists $\delta>0$ such that for every
 $\phi_\pm\in\Sigma_{\alpha}$ with $|||\phi_\pm|||_{\alpha}<\delta$, the
equations~\eqref{NLSch4-} and~\eqref{NLSch4+} have unique globally
defined continuous solutions $\phi^\pm:\mathbb R\to \Sigma_{\alpha}$ fulfilling
\begin{gather}\label{solsmall1}
\Vert\phi^\pm(t)- e^{-it{{\bf\Delta}_\alpha}}\phi_\pm\Vert_{2,\alpha}  \to
0,\qquad\mbox{for}\ \ t\to\pm\infty .
\end{gather}
If, furthermore, $F$ has no linear term, then we have
\begin{gather}\label{solsmall2}
 |||e^{it{{\bf\Delta}_\alpha}}\phi^\pm(t)- \phi_\pm |||_{\alpha}  \to
0,\qquad\mbox{for}\  \ t\to\pm\infty .
\end{gather}
\end{theorem}

\begin{remark} We do not know at present whether the assumption that $F$
be without linear term is strictly necessary or is merely an artifact
due to the crudeness of the methods applied. A~similar, but weaker,
limitation on the behaviour of $F$ close to $0$
appears in the corresponding result for the classical NLS quoted in
\cite[Teorem 6.6]{Strauss}, where the specif\/ication of the scattering
spaces is, however, not made explicit.
\end{remark}

This result allows the def\/inition of injective {\it wave operators}
$\Omega_{\pm}: \{\phi_\pm\in\Sigma_{\alpha}|\;
|||\phi_\pm|||_\alpha\leq\delta\} \to\Sigma_{\alpha}$ for small data
at $\pm\infty$ in the standard fashion by
\begin{gather}\label{waveop}
\Omega_{\pm}\phi_\pm = \phi^\pm(0) .
\end{gather}
It even allows a def\/inition of a {\it scattering operator}
$S=\Omega_+^{-1}\Omega_-$ for suf\/f\/iciently small data at~$-\infty$,
see Remark~\ref{remscat} below. Existence of wave operators def\/ined
for arbitrary data in
$\Sigma_\pm$ would follow if the corresponding Cauchy problem
\begin{gather*}
\phi(t) = e^{-i(t-t_0){\bf\Delta}_\alpha}\phi_0 - i\int_{t_0}^t
e^{i(s-t)t{\bf\Delta}_\alpha}\phi(s) F\big(|\phi(s)|^2\big)\,ds ,
\end{gather*}
 has global solutions for all $\psi_0\in{\cal H}_\alpha$. The proof of
the global existence result of Theorem~\ref{th1} relies on the
conservation of $\Vert\cdot\Vert_2$-norm, which does not hold for the
$\Vert\cdot\Vert_{2,\alpha}$-norm if $\alpha\neq 0$. Hence, we do not at
present know how to treat large scattering data except for the case where
 $\phi$ is assumed to be a diagonal operator w.r.t. the harmonic
oscillator basis $\{|n\rangle\}$, and in particular for the
rotationally symmetric case. The results for this case are reported in
 Section~\ref{secdiag}.

\section{Existence of global solutions}
\label{secth1}

In this section we give a proof of Theorem \ref{th1}.

\begin{proof}[Proof of part $\boldsymbol{a)}$.]  This follows by a straight-forward
application of well known techniques, see e.g.~\cite{Reed}. Hence, we
only indicate the main line of argument.

Iterating the inequality
\[
\Vert \phi_1\phi_2- \psi_1\psi_2\Vert_2 \leq
\Vert\phi_1-\psi_1\Vert_2\Vert\phi_2\Vert_2+\Vert\psi_1\Vert_2\Vert\phi_2-\psi_2\Vert_2,\qquad
 \phi_1, \phi_2, \psi_1, \psi_2\in{\cal H},
\]
and using $\Vert\phi\Vert _2=\Vert \phi^*\Vert_2$ we obtain
\begin{gather}\label{H1}
\Vert\phi F(\phi^*\phi)-\psi F(\psi^*\psi)\Vert_2 \leq C_1(
\Vert\phi\Vert_2,\Vert\psi\Vert_2)\Vert\phi-\psi\Vert_2  ,
\end{gather}
where $C_1$ is a polynomial with positive coef\/f\/icients, in particular an
increasing function of $\Vert\phi\Vert_2$ and $\Vert\psi\Vert_2$.
By Corollary~1 to Theorem~1 of \cite{Reed}, this suf\/f\/ices to ensure existence of a
local continuous solution $\phi: \,]t_0-T,t_0+T[\, \to {\cal H}$ to (\ref{NLSch3}), where $T$
is a decreasing function of $\Vert\phi_0\Vert_2$.

Using that the mappings $\phi\to [a_k,\phi]$ and  $\phi\to
[a^*_k,\phi]$ are derivations on $W^{-1}{\cal S}(\mathbb R^{2d})$, it
follows that ${{\bf\Delta}}(\phi_1\cdots\phi_n)$ is a polynomial in
$\phi_i$, $[a_k,\phi_i]$, $[a_k^*,\phi_i]$,
${\bf\Delta}\phi_i$, $i=1,\dots,n$, $k=1,\dots,d$. Since
\begin{gather}\label{est1}
\Vert[a_k,\phi]\Vert^2_2 =
\mbox{Tr}(\phi^*[a_k^*,[a_k,\phi]])=
\mbox{Tr}(\phi^*{{\bf\Delta}}\phi)\leq \Vert\phi\Vert_2\Vert{{\bf\Delta}}\phi\Vert_2 ,
\end{gather}
we conclude as above that
\begin{gather*}
\Vert {\bf\Delta}(\phi F(\phi^*\phi)-\psi F(\psi^*\psi))\Vert_2 \leq
C_2(\Vert\phi\Vert_2,\Vert\psi\Vert_2,\Vert{\bf\Delta}\phi\Vert_2,\Vert{\bf\Delta}\psi\Vert_2)(\Vert\phi-\psi\Vert_2
+ \Vert{\bf\Delta}(\phi-\psi)\Vert_2),
\end{gather*}
where $C_2$ is an increasing function of its arguments. In the f\/irst
place, this inequality holds for $\phi,\psi\in W^{-1}{\cal S}(\mathbb
R^{2d})$, but since ${\bf\Delta}$ equals the closure of its restriction to $W^{-1}{\cal S}(\mathbb R^{2d})$,
it holds for all $\phi$, $\psi$ in its domain $D({\bf\Delta})$. By Theorem~1 in~\cite{Reed} (or rather its proof) it follows that for each $\phi_0\in
D({\bf\Delta})$ there exists a unique strong solution
$\phi: \; ]t_0-T,t_0+T[\to D({\bf\Delta})$ to equation~(\ref{NLSch2}) with $\phi(t_0)=\phi_0$, where $T>0$ can be
chosen as a decreasing function of $\Vert\phi_0\Vert_2$ and
$\Vert{\bf\Delta}\phi_0\Vert_2$.

For strong solutions the 2-norm is conserved:
\begin{gather}\label{conserv1}
\frac{d}{dt}\Vert\phi(t)\Vert^2_2
= -i{\rm Tr}\big[({\bf\Delta}\phi -\phi F(\phi^*\phi))^*\phi\big] +
i{\rm Tr}\big[\phi^*({\bf\Delta}\phi-\phi F(\phi^*\phi))\big] =0 .
\end{gather}
Furthermore, since the polynomial ${\bf\Delta}(\phi F(\phi^*\phi))$ considered
above,  is a sum of monomials each of which either contains one factor
${\bf\Delta}\phi$ or ${\bf\Delta}\phi^*$ or two factors of the form
$[a_k,\phi]$, $[a^*_k,\phi]$, $[a_k,\phi^*]$,  $[a^*_k,\phi^*]$, we
conclude from (\ref{est1}) that
\begin{gather}\label{H02}
\Vert{\bf\Delta}( \phi F(\phi^*\phi))\Vert _2\leq C_3(\Vert\phi\Vert_2)
\Vert{\bf\Delta}\phi\Vert_2 ,\qquad \phi\in D({\bf\Delta}) ,
\end{gather}
where $C_3$ is an increasing function of its argument.

Using (\ref{conserv1}) and (\ref{H02}) it follows from Theorem 2 of
\cite{Reed} that strong solutions are globally def\/ined, i.e. we can
choose $T=\infty$. Finally, using (\ref{H1}) and (\ref{conserv1}) again
one shows via Corollary~2 to Theorem~14 of \cite{Reed} that the weak
solutions are likewise globally def\/ined. This proves $a)$.
\end{proof}

\begin{proof}[Proof of part $\boldsymbol{b)}$.] Let $G$ be the polynomial vanishing
at $x=0$ and satisfying $G'=F$. Under the stated assumptions on $F$
it follows that the polynomial $G(x^2)-F(x_0)x^2$ is positive  except
at $x=0$, which is a second order zero, and its derivative is positive
on $\mathbb R_+$ except for a zero at $x=\sqrt{x_0}$. Hence, for
$\epsilon >0$ suf\/f\/iciently small, the polynomial
\[
V(x) = \frac 12\big(G\big(x^2\big)- (F(x_0)+\epsilon)x^2\big) ,
\]
is positive, except for a second order zero at $x=0$, and has a
a single local minimum on $\mathbb R_+$. Thus $V(x)$ fulf\/ills the
assumptions of Theorem~1 of~\cite{DJN} implying the existence of a
self-adjoint solution $\phi_0\in D({\bf\Delta})$ to the equation
\[
{\bf\Delta}\phi + \theta V'(\phi) = 0 ,
\]
if $\theta$ is suf\/f\/iciently large.  It then follows that $\phi(t) =
e^{i\omega(t_0-t)}\phi_0$ is a strong solution of~(\ref{NLSch2}), with
$\omega=\theta(F(x_0)+\epsilon)$.
This proves $b)$.
\end{proof}

\section{Some norm inequalities}
\label{secframe}

In preparation for the proofs of Theorems \ref{th2} and \ref{th3} we
collect in this section a few useful lemmas.

\begin{lemma}
\label{Bbounded}
For $\alpha<-d$, the matrix
$\{b_{mn}^\alpha\}$ represents a bounded operator $B_\alpha$ on
$L^2(\mathbb R^d)$ with respect to the basis $\{|n\rangle\}$.
\end{lemma}
\begin{proof}
For any $x,y\in L^2(\mathbb R^d)$, consider their coordinate sequences
$\{x_n\},\{y_n\}\in \ell^2(\mathbb N_0^d)$ w.r.t.\ the basis $\{|n\rangle\}$. We have
\begin{gather*}
\vert\langle x,B_\alpha y\rangle\vert=
\bigg|\sum_{m,n}\overline{x_m} b^\alpha_{mn} y_n\bigg| \leq
\sum_{m,n}(1+|m-n|)^\alpha |x_m| |y_n|\\
\hphantom{\vert\langle x,B_\alpha y\rangle\vert=
\bigg|\sum_{m,n}\overline{x_m} b^\alpha_{mn} y_n\bigg|}{}
\leq\sum_{m,k}(1+|k|)^\alpha |x_m| |y_{k+m}|
\leq\sum_k(1+|k|)^\alpha \Vert x\Vert_{L^2} \Vert y\Vert_{L^2} ,
\end{gather*}
which concludes the proof.
\end{proof}

\begin{lemma}\label{anorm}
For any $\alpha>d$, there exists a constant $C_\alpha>0$ such that the following holds:
\[
\Vert\phi\Vert_a\leq C_\alpha \Vert\phi\Vert_{\infty,\alpha} ,\qquad \phi\in{\cal L}_{\infty,\alpha} .
\]
\end{lemma}

\begin{proof}
With notation as in the previous proof we have, for $\phi\in{\cal L}_{\infty,\alpha}$,
\begin{gather*}
\bigg|\sum_{m,n}|\phi_{mn}| \overline{x_n}\,y_m\bigg| \leq
\sum_{m,n}\big|b^\alpha_{mn} \phi_{mn}\big|
\big|b^{-\alpha}_{mn} \overline{x_n} y_m\big|
\leq \sup_{m,n}\big|b^\alpha_{mn}\phi_{mn}\big|
 \Vert B_{-\alpha}\Vert_{\rm op} \Vert x\Vert_{L^2} \Vert y\Vert_{L^2} ,
\end{gather*}
which evidently implies the claim by Lemma~\ref{Bbounded}.
\end{proof}

\begin{lemma}\label{trace}
For any $\alpha<-d$ there exists a constant $C'_\alpha>0$ such that
\[
\Vert\phi\Vert_{1,\alpha}\leq C'_{\alpha} \Vert\phi\Vert_1 ,\qquad \phi\in{\cal L}_1 ,
\]
where ${\cal L}_1$ denotes the space of trace class operators and
$\Vert\cdot\Vert_1$ is the standard trace-norm.
\end{lemma}

\begin{proof}
Writing the trace class operator $\phi$ as $\frac
12(\phi+\phi^*)+\frac i2(i\phi^*-i\phi)$ we may assume that $\phi$ is
self-adjoint. Then, writing $\phi= \phi_+-\phi_-$ where $\phi_\pm$ are positive
operators each of which has trace norm at most that of $\phi$, we can
assume $\phi$ is positive. In this case the Cauchy--Schwarz inequality
gives
\[
|\phi_{kl}|  \leq \big(\phi_{kk} \phi_{ll}\big)^{1/2},
\]
and hence
\begin{gather*}
\sum_{k,l}b_{kl}^{\alpha} |\phi_{kl}|\leq
\sum_{k,l}b_{kl}^{\alpha} \phi_{kk}^{1/2} \phi_{ll}^{1/2}
\leq \Vert B_{\alpha}\Vert_{\rm op} \sum_k\phi_{kk}
=\Vert B_{\alpha}\Vert_{\rm op} \Vert\phi\Vert_1 ,
\end{gather*}
which proves the claim thanks to Lemma~\ref{Bbounded}.
\end{proof}

\begin{lemma}\label{product}
$a)$  For any $\alpha\geq 0$ there exists a constant $C_{1,\alpha}>0$
such that
\begin{gather}\label{prodest1}
\Vert\phi \psi\Vert_{2,\alpha}\leq C_{1,\alpha}(\Vert\phi\Vert_{2,\alpha} \Vert\psi\Vert_a
+\Vert\phi\Vert_a \Vert\psi\Vert_{2,\alpha}),\qquad \phi, \psi\in {\cal L}_{2,\alpha} .
\end{gather}

$b)$  For any $\alpha> d$ there exists a constant $C_{2,\alpha}>0$
such that
\begin{gather}\label{prodest2}
\Vert\phi \psi\Vert_{1,\alpha}\leq C_{2,\alpha}(\Vert\phi\Vert_{2,2\alpha} \Vert\psi\Vert_2
+\Vert\phi\Vert_2 \Vert\psi\Vert_{2,2\alpha}),\qquad \phi,\psi\in {\cal L}_{2, 2\alpha} .
\end{gather}
\end{lemma}

\begin{proof}
First, note that for $\alpha\geq 0$
\begin{gather}
\label{Bsum}
b_{mn}^\alpha\leq c(\alpha)\,(b_{mk}^\alpha+b_{kn}^\alpha) ,
\qquad  m,n,k\in\mathbb N_0^d,
\end{gather}
where the constant $c(\alpha)$ depends only on $\alpha$.
We then have
\begin{gather*}
\Vert\phi \psi\Vert_{2,\alpha}^2 =\sum_{m,n}b^{2\alpha}_{mn}\bigg|\sum_k\phi_{mk}
\psi_{kn}\bigg|^2
 \leq
C\sum_{m,n}\bigg(\sum_kb^{\alpha}_{mk} |\phi_{mk}| \, |\psi_{kn}|
+\sum_kb^{\alpha}_{kn} |\phi_{mk}|\,|\psi_{kn}|\bigg)^2\\
\hphantom{\Vert\phi \psi\Vert_{2,\alpha}^2}{}
\leq
2C\sum_{m,n}\bigg(\sum_kb^{\alpha}_{mk} |\phi_{mk}|\,|\psi_{kn}|\bigg)^2
+2C\sum_{m,n}\bigg(\sum_kb^{\alpha}_{kn} |\phi_{mk}|\,|\psi_{kn}|\bigg)^2\\
\hphantom{\Vert\phi \psi\Vert_{2,\alpha}^2}{}
=2C\,\big(\Vert(U_{\alpha}\widetilde\phi) \widetilde\psi\Vert_2^2+
\Vert\widetilde\phi (U_{\alpha}\widetilde\psi)\Vert_2^2\big)\\
\hphantom{\Vert\phi \psi\Vert_{2,\alpha}^2}{}
\leq 2C \big(\Vert U_{\alpha}\widetilde\phi\Vert_2^2
\Vert\widetilde\psi\Vert_{\rm op}^2+
\Vert U_\alpha\widetilde\psi\Vert_2^2
\Vert\widetilde\phi\Vert_{\rm op}^2\big)
 \leq 2C \big(\Vert\phi\Vert_{2,\alpha}
\Vert\psi\Vert_a+\Vert\psi\Vert_{2,\alpha} \Vert\phi\Vert_a\big)^2 ,
\end{gather*}
where  $C= c(\alpha)^2$ and $U_\alpha:\mathcal H_\alpha\to\mathcal H_0$ is
the unitary operator def\/ined by (\ref{Ualpha}). This establishes~(\ref{prodest1}).

Using again~\eqref{Bsum}, we have
\begin{gather}
\Vert\phi \psi\Vert_{1,\alpha} =\sum_{m,n}b^\alpha_{mn}\bigg|\sum_k\phi_{mk}
\psi_{kn}\bigg| \leq
c(2\alpha)\sum_{m,n,k}b^{-\alpha}_{mn}\big(|b^{2\alpha}_{mk} \phi_{mk}|\,
|\psi_{kn}|+|\phi_{mk}|\,|b^{2\alpha}_{kn} \psi_{kn}|\big)\nonumber\\
\hphantom{\Vert\phi \psi\Vert_{1,\alpha}}{}
=c(2\alpha)\big(\Vert(U_{2\alpha}\widetilde\phi) \widetilde\psi\Vert_{1,-\alpha}+
\Vert\widetilde\phi (U_{2\alpha}\widetilde\psi)\Vert_{1,-\alpha}\big) .\label{34}
\end{gather}
 Lemma~\ref{trace} and a H\"older inequality now yield, for $\alpha>d$,
\[
\Vert(U_{2\alpha}\widetilde\phi) \widetilde\psi\Vert_{1,-\alpha}
\leq C'_{-\alpha}\Vert(U_{2\alpha}\widetilde\phi) \widetilde\psi\Vert_1
\leq C'_{-\alpha}\Vert U_{2\alpha}\widetilde\phi\Vert_2\Vert\psi\Vert_2
= C'_{-\alpha}\Vert\phi\Vert_{2,2\alpha}\Vert\psi\Vert_2 .
\]
Combining this with \eqref{34} gives \eqref{prodest2}.
\end{proof}

\section{Decay of free solutions}
\label{secth2}

The goal in this section is to prove Theorem \ref{th2}. The main
step in the proof is to establish appropriate time-decay estimates for the matrix elements
$\left(e^{-it{\bf\Delta}}\right)_{nm,kl}$ of the free time-evolution
operator with respect to the orthonormal basis
$\{(|n\rangle\langle m|)\}$ for ${\cal H}$. This is our f\/irst objective.

Since the Weyl map $W$ is unitary up to a constant factor by \eqref{Wunitary}, the matrix elements of the heat operator $e^{-t{\bf\Delta}}$, $t>0$, can be
obtained in closed form from the well-known  expression for
the heat kernel in Euclidean space and the relation \eqref{tria},
making use of the fact that the functions $W^{-1}(|n\rangle\langle
m|)$ can be computed explicitly in terms of Laguerre polynomials. The
result is given in \cite{Raimar} and reads, in case $d=1$,
\begin{gather*}
\left(e^{-t{\bf\Delta}}\right)_{nm,kl}=
\delta_{m+k,n+l}\sum_{v=0}^{\min\{m,l\}} C_{nm,kl,v}
 \frac{t^{m+l-2v}}{(1+t)^{m+k+d}} ,
\end{gather*}
where
\begin{gather}
\label{matconst}
C_{nm,kl,v}:=
\sqrt{{n\choose m-v}{k\choose l-v}
{m\choose m-v}{l\choose l-v}} ,
\end{gather}
for arbitrary non-negative integers $n$, $m$, $k$, $l$. By convention the
binomial coef\/f\/icient $n\choose m$ vanishes unless $0\leq m\leq n$. Note that the presence of
the Kronecker delta factor is a consequence of rotational invariance
of the heat kernel in two-dimensional space, which in the operator
formulation entails that $ e^{-t{\bf\Delta}}$ commutes with $e^{-isa^*a}$ for
$s\in \mathbb R$, where we have dropped the subscript on the
annihilation and creation operators when $d=1$.

Since ${\bf\Delta}$ is a positive,  self-adjoint operator we obtain
$\left(e^{-it{\bf\Delta}}\right)_{nm,kl}$ for $t\in\mathbb R$ by analytic
continuation in $t$, that is
\begin{gather}
\label{mateluni}
\left(e^{-it{\bf\Delta}}\right)_{nm,kl}=
\delta_{m+k,n+l}\sum_{v=0}^{\min\{m,l\}} C_{nm,kl,v}
 \frac{(it)^{m+l-2v}}{(1+it)^{m+k+d}} .
\end{gather}

We next observe that these  matrix
elements are expressible in terms of Jacobi polynomials.

\begin{lemma}
\label{jacobi}
For $d=1$ and $l\leq m,k\leq n$, we have
\begin{gather}\label{jacrep}
\left(e^{-it{\bf\Delta}}\right)_{nm,kl}=
\delta_{m+k,n+l}
\sqrt{\frac{n! l!}{m! k!}}
\frac{(it)^{m+l}}{(1+it)^{m+k+1}}\big(1+t^{-2}\big)^l
  P_{l}^{n-m,m-l}\left(\frac{t^2-1}{t^2+1}\right) ,
\end{gather}
where $P_l^{\alpha,\beta}(X)$, $l\in\mathbb N_0$, $\alpha,\beta\geq 0$,
$X\in[-1,1]$, are the classical Jacobi polynomials with standard
normalization~{\rm \cite{Szego}}.
\end{lemma}

\begin{proof}
For  $m+k=n+l$ and $n\geq m$ we have
\begin{gather*}
C_{nm,kl,v}
=\sqrt{\frac{n! l!}{m! k!}}{m\choose v}{l+n-m\choose l-v},
\end{gather*}
and consequently, for $l\leq m$,
\begin{gather}
\label{matelements}
\left(e^{-it{\bf\Delta}}\right)_{nm,kl}=
\delta_{m+k,n+l}\sqrt{\frac{n! l!}{m! k!}}
\frac{(it)^{m+l}}{(1+it)^{m+k+1}}
\sum_{v=0}^{l} {m\choose v}{l+n-m\choose l-v}
 \big(-t^{-2}\big)^{v} .
\end{gather}
Recall now \cite{Szego} that the classical Jacobi polynomials
$P_l^{\alpha,\beta}(X)$, $l\in\mathbb N_0,$ are orthogonal
w.r.t.\ the weight function
\begin{gather}\label{weight}
w(X)   =  (1-X)^\alpha(X+1)^\beta,\qquad X\in[-1,1],
\end{gather}
for f\/ixed values of $\alpha, \beta > -1$. For our purposes we may
restrict attention to integer values of $\alpha$ and  $\beta$. A  convenient
explicit form of $P_l^{\alpha,\beta}(X)$ with standard normalization is
\[
P_l^{\alpha,\beta}(X)=\sum_{j=0}^l{l+\alpha\choose l-j}{l+\beta\choose
  j}\left(\frac{X-1}{2}\right)^j\left(\frac{X+1}{2}\right)^{l-j} .
\]
With
\begin{gather}\label{X-t}
X=(t^2-1)/(t^2+1)\in [-1,1]
\end{gather}
this gives
\[
P_l^{\alpha,\beta}\left(\frac{t^2-1}{t^2+1}\right)
= \big(1+t^{-2}\big)^l\sum_{j=0}^l{l+\alpha\choose l-j}
{l+\beta\choose j}\big(-t^{-2}\big)^{j}
\]
and the result follows by comparison with~\eqref{matelements}.
\end{proof}

The representation~\eqref{jacrep} combined with rather recent
uniform estimates on the Jacobi polynomials may now be used to derive the
following bound on the matrix elements in question.

\begin{lemma}
\label{Uniest1}
For $d\geq 1$, there exists a constant $C_d$, independent of $n,m,k,l
\in\mathbb N_0^d$ and
$t\in\mathbb R$, such that
\begin{gather}\label{decay0}
\big|\left(e^{-it{\bf\Delta}}\right)_{nm,kl}\big| \leq  C_d |t|^{-\frac
d2} ,\qquad |t|\geq 1 .
\end{gather}
\end{lemma}

\begin{proof} From the def\/inition of ${\bf\Delta}$ it follows that the
  matrix elements in question factorize into those for the
  one-dimensional case, so it suf\/f\/ices to consider $d=1$.

 Since $e^{-it{\bf\Delta}}$ is unitary, it follows from
  \eqref{mateluni} that
 \[
\left(e^{-it{\bf\Delta}}\right)_{nm,kl} =  \left(e^{-it{\bf\Delta}}\right)_{kl,nm}.
\]
 Moreover, since
$W(\overline f)=W(f)^*$ and the heat kernel on $\mathbb R^2$ is
symmetric in its two arguments we likewise have
 \[
\left(e^{-it{\bf\Delta}}\right)_{nm,kl} =  \left(e^{-it{\bf\Delta}}\right)_{mn,lk} .
\]
These symmetries may also be checked directly from \eqref{mateluni} and \eqref{matconst}.

Taking into account that $m+k=n+l$ for non-vanishing matrix elements
we can therefore assume that $l\leq m,k\leq n$. Lemma~\ref{jacobi} then yields
\begin{gather}
\label{estimategene}
\left|\left(e^{-it{\bf\Delta}}\right)_{nm,kl}\right|
 = \delta_{m+k,n+l}
\sqrt{\frac{n! l!}{m! k!}}(1+t^2)^{-1/2}
\big(1+t^{-2}\big)^{(l-n)/2}|t|^{l-k}
\left|P_l^{\alpha,\beta}\left(\frac{t^2-1}{t^2+1}\right)\right|,
\end{gather}
where we have set
\begin{gather*}
\alpha:=n-m \qquad\mbox{and}\qquad \beta:=m-l .
\end{gather*}
Introducing the orthonormal Jacobi polynomials ${\bf
  P}_l^{\alpha,\beta}$ of degree $l\geq 0$ by normalizing w.r.t.\ the
$L^2$-norm def\/ined by the weight \eqref{weight}, one f\/inds (see e.g.~\cite[p.~67]{Szego})
\[
{\bf P}_l^{\alpha,\beta}(X)  =  \sqrt{\frac{(2l+\alpha +\beta +1)}{2^{\alpha
    +\beta +1}}\frac{(l+\alpha +\beta)! l!}{(l+\alpha)! (l+\beta)!}}
P_l^{\alpha,\beta}(X) ,
\]
and \eqref{estimategene} can be rewritten as
\begin{gather}
\left|\left(e^{-it{\bf\Delta}}\right)_{nm,kl}\right|
 = \delta_{m+k,n+l}\left(l+\frac{\alpha +\beta +1}2\right)^{-\frac 12}\nonumber\\
 \hphantom{\left|\left(e^{-it{\bf\Delta}}\right)_{nm,kl}\right|=}{}\times
\big(1+t^2\big)^{-1/2} (1-X)^{\alpha/2}(1+X)^{\beta/2}
\big|{\bf P}_l^{\alpha,\beta}(X)\big| ,\label{estimategene2}
\end{gather}
with $X$ given by \eqref{X-t}.

We now use the following uniform bound for the orthonormal Jacobi polynomials,
proven in~\cite{erdel}:
\begin{gather}\label{erdelyi}
(1-X)^{\alpha/2+1/4}
(1+X)^{\beta/2+1/4}\big|{\bf P}_l^{\alpha,\beta}(X)\big| \leq \sqrt{\frac{2 e}{\pi}}\sqrt{2+
\sqrt{\alpha^2+\beta^2}} ,
\end{gather}
valid for all $X\in]-1,1[$, $l\geq 0 $ and
$\alpha, \beta\geq-1/2$.

Applying this inequality in conjunction with \eqref{estimategene2} we
obtain, for $m,n,k,l\geq 0$ and   $|t|\geq 1$,
\begin{gather*}
\left|\left(e^{-it{\bf\Delta}}\right)_{mn,kl}\right|
 \leq C \delta_{m+k,n+l}
\left(\frac{2+
\sqrt{\alpha^2+\beta^2}}
{2l+\alpha +\beta +1}\right)^{\frac 12}
|t|^{-\frac 12}\leq C' \delta_{m+k,n+l} |t|^{-\frac 12} ,
\end{gather*}
which is the announced result.
\end{proof}

In a recent article \cite{Krasikov}, Krasikov has improved the bound
\eqref{erdelyi} to
\begin{gather*}
(1-X)^{\alpha/2+1/4}
(1+X)^{\beta/2+1/4}\big|{\bf P}_l^{\alpha,\beta}(X)\big| \leq \sqrt 3
 \alpha^{1/6}\left(1+\frac\alpha l\right)^{1/12} ,
\end{gather*}
valid for $X\in]-1,1[$, $l\geq 6$ and
$\alpha\geq \beta\geq(1+\sqrt2)/4$. Apart from the restricted domain
of validity, which presumably can be extended to $\alpha, \beta, l\geq 0$,
it is not clear whether it may lead to a stronger decay estimate than~\eqref{decay1}. In particular, we do not know whether the decay exponent~$d/2$ in~\eqref{decay0} is optimal.
It is, on the other hand, easy to obtain improved uniform bounds on
diagonal matrix elements as demonstrated by the following
Lemma, to be used in the treatment of the diagonal case in
Section~\ref{secdiag}.

\begin{lemma}
\label{logestimate}
For any $d\geq 1$ there exists a constant $C'_d$, independent of $n,m
\in\mathbb N_0^d$ and
$t$, such that
\begin{gather}\label{decay00}
|\left(e^{-it{\bf\Delta}}\right)_{nn,mm}| \leq  C'_d |t|^{-d}(1+\log
|t|)^d ,\qquad |t|\geq 1 .
\end{gather}
\end{lemma}

\begin{proof}
Again, we may assume that $d=1$. We then obtain from \eqref{estimategene}
\[
\left|\left(e^{-it{\bf\Delta}}\right)_{nn,mm}\right|
 = (1+t^2)^{-1/2}
\big(1+t^{-2}\big)^{-\beta/2}
\left|P_l^{0,\beta}\left(\frac{t^2-1}{t^2+1}\right)\right|,\qquad \mbox{with} \quad \beta=n-m ,
\]
and \eqref{decay00} will follow once we show that
there exist constants $C_1$, $C_2$, independent of $\alpha,\beta, l\in\mathbb
N_0$ and $X\in]-1,1[$, such that
\[
\big|P_l^{\alpha,\beta}(X)\big|\leq  \left(\frac2{1-X}\right)^{\alpha/2}
\left(\frac2{1+X}\right)^{\beta/2} \big(C_1 +
C_2\big|\log(1-|X|)\big|\big)  .
\]
This alternative estimate on the Jacobi
polynomials relies on a very standard integral representation. We give a
detailed proof for the sake of completeness.
Taking into account the symmetry
relation
\begin{gather*}
P_l^{\alpha,\beta}(X)  =  (-1)^l\,P_l^{\beta,\alpha}(-X) ,
\end{gather*}
 it suf\/f\/ices to
  consider $X\in [0,1[$.
Setting $X=\cos\theta$, $\theta\in\, ]0,\frac{\pi}{2}]$,
the assertion we want to prove is equivalent to
\[
\left(\sin\frac\theta2\right)^\alpha \left(\cos\frac\theta2\right)^\beta
\big|P_l^{\alpha,\beta}(\cos\theta)\big|\ \leq  C_1 + C_2 |\log\theta | .
\]
For $X\ne\pm1$, we can use the following integral representation, see e.g.~\cite[p. 70]{Szego}:
\[
P_l^{\alpha,\beta}(X)=\frac{2^{\alpha+\beta}}{2\pi i}
\int_\Gamma\frac{dz}{z^{l+1}} \frac{R(X,z)^{-1}}
{\big(1-z+R(X,z)\big)^\alpha\big(1+z+R(X,z)\big)^\beta} ,
\]
where $R(X,z)=\sqrt{1-2Xz+z^2}$ and $\Gamma$ is a positively oriented
contour in the complex plane enclosing the point $X$. Choosing
for $\Gamma$ the unit circle centered at the origin, this yields
\[
P_l^{\alpha,\beta}(\cos\theta)=\frac{2^{\alpha+\beta}}{2\pi}
\int_{-\pi}^\pi d\varphi \frac{e^{-il\varphi} R(\theta,\varphi)^{-1}}
{\big(1-e^{i\varphi}+R(\theta,\varphi)\big)^\alpha\big(1+e^{i\varphi}+R(\theta,\varphi)\big)^\beta} ,
\]
where
\[
R(\theta,\varphi)=\sqrt{1-2\cos\theta e^{i\varphi}+e^{2i\varphi}}=
2\sqrt{-e^{i\varphi} \sin\frac{\varphi+\theta}2 \sin\frac{\varphi-\theta}2}
\]
and the complex square root is def\/ined in such a way  that
$\sqrt1=1$. Using a trigonometric identity one f\/inds
\[
R(\theta,\varphi)=\begin{cases}
\displaystyle 2 e^{i\varphi/2} \left|\sin\frac{\varphi+\theta}2 \sin\frac{\varphi-\theta}2\right|^{1/2}, \quad  &
\displaystyle  \mbox{if}\quad \left|\sin\frac\varphi2\right|\leq\sin\frac\theta2 ,\vspace{1mm}\\
\displaystyle  -2i e^{i\varphi/2} \left|\sin\frac{\varphi+\theta}2 \sin\frac{\varphi-\theta}2\right|^{1/2},\quad  &
\displaystyle  \mbox{if}\quad \sin\frac\varphi2\geq\sin\frac\theta2 ,\vspace{1mm}\\
\displaystyle  2i e^{i\varphi/2} \left|\sin\frac{\varphi+\theta}2 \sin\frac{\varphi-\theta}2\right|^{1/2},\quad &
\displaystyle   \mbox{if}\quad \sin\frac\varphi2\leq-\sin\frac\theta2 ,
\end{cases}
\]
and thus
\begin{gather*}
\big|1-e^{i\varphi}+R(\theta,\varphi)|
 =2\left|\sin\frac\varphi2+\frac i2 e^{-i\varphi/2} R(\theta,\varphi)\right|\\
\qquad =\begin{cases}
\displaystyle 2\left(\sin^2\frac\varphi2+\left|\sin\frac{\varphi+\theta}2
\sin\frac{\varphi-\theta}2\right|\right)^{1/2}=2\sin\frac\theta2, \quad&
\displaystyle  \mbox{if}\quad
\left|\sin\frac\varphi2\right|\leq\sin\frac\theta2 ,\vspace{1mm}\\
\displaystyle  2\left(\left|\sin\frac\varphi2\right|+\left|\sin\frac{\varphi+\theta}2
\sin\frac{\varphi-\theta}2\right|\right)\geq2\sin\frac\theta2, \quad &
\displaystyle  \mbox{if}\quad
\left|\sin\frac\varphi2\right|\geq\sin\frac\theta2 .
\end{cases}
\end{gather*}
Similarly, one establishes
\[
\big|1+e^{i\varphi}+R(\theta,\varphi)|\geq2\left|\cos\frac\theta2\right| .
\]
This eventually implies the following bound:
\begin{gather*}
\left(\sin\frac\theta2\right)^\alpha \left(\cos\frac\theta2\right)^\beta
\big|P_l^{\alpha,\beta}(X)\big|\\
\qquad{} \leq \frac1{2\pi}\int_{-\pi}^\pi d\varphi
\big|R(\theta,\varphi)\big|^{-1}
=\frac1{2\pi}\int_0^\pi d\varphi
\left|\sin\frac{\varphi+\theta}2 \sin\frac{\varphi-\theta}2\right|^{-1/2}.
\end{gather*}
Finally, it is easily seen that there exists an upper bound of the
form $ C_1 + C_2 |\log\theta |$ for the latter integral, thus
completing the proof.
\end{proof}

We are now ready to prove Theorem~\ref{th2}.

\begin{proof}[Proof of Theorem~\ref{th2}]
Recalling the def\/inition \eqref{Ualpha} of $U_\alpha$ and that
$b_{mn}$ only depends on $|n-m|$, it follows from the presence of the
Kronecker-delta factor in $(e^{-it{\bf{\Delta}}})_{nm,kl}$ that
\[
e^{-it{\bf{\Delta}}_\alpha}(|n\rangle\langle m|) = U_\alpha^{-1}
e^{-it{\bf{\Delta}}}U_\alpha(|n\rangle\langle m|) = e^{-it{\bf{\Delta}}}(|n\rangle\langle m|)
\]
is independent of $\alpha$. Since
$\{b_{mn}^{-\alpha}|n\rangle\langle m|\}$ is an orthonormal basis
for ${\cal H}_\alpha$ we obtain, for $\alpha\geq 0$ and
\[
\phi =
\sum_{m,n}\phi_{mn}|n\rangle\langle m| \in {\cal H}_\alpha\subseteq
{\cal H} ,
\]
that
\[
e^{-it{\bf{\Delta}}_\alpha}\phi =
\sum_{n,m}\phi_{mn}e^{-it{\bf{\Delta}}_\alpha}(|n\rangle\langle m|) =
\sum_{n,m}\phi_{mn}e^{-it{\bf{\Delta}}}(|n\rangle\langle m|) = e^{-it{\bf{\Delta}}}\phi ,
\]
i.e. $e^{-it{\bf{\Delta}}_\alpha}$ equals the restriction of $e^{-it{\bf{\Delta}}}$
to ${\cal H}_\alpha$.
Applying  Lemma~\ref{anorm} we then get, for $\beta>d$ and
$\phi\in {\cal H}_\alpha\cap {\cal L}_{1,\beta}$,
\begin{gather*}
\big\Vert e^{-it{\bf{\Delta}}_\alpha}\phi\big\Vert_a
 \leq
 C_\beta \sup_{n,m}\bigg|b_{nm}^{\beta}
\sum_{k,l}\big(e^{-it{\bf{\Delta}}}\big)_{nm,kl}\,\phi_{kl}\bigg|
 \leq C_\beta \sup_{n,m}\sum_{k,l}\Big|
\big(e^{-it{\bf{\Delta}}}\big)_{nm,kl} b_{kl}^{\beta} \phi_{kl}\Big|\\
\phantom{\big\Vert e^{-it{\bf{\Delta}}_\alpha}\phi\big\Vert_a}{}
\leq
C_\beta \sup_{n,m,k,l}\Big|\big(e^{-it{\bf{\Delta}}}\big)_{nm,kl}\Big| \sum_{k,l}\big|b_{kl}^{\beta} \phi_{kl}\big|
= C_\beta \sup_{n,m,k,l}\Big|\big(e^{-it{\bf{\Delta}}}\big)_{nm,kl}\Big|
\Vert\phi\Vert_{1,\beta} ,
\end{gather*}
where, in the second step, we have once more made use of the
Kronecker-delta factor in $(e^{-it{\bf{\Delta}}})_{nm,kl}$.
The claim now follows from Lemma~\ref{Uniest1}.
\end{proof}

\begin{remark}
The time-decay exponent $\frac d2$ found in  Lemma~\ref{Uniest1} equals half
the value of the one for $2d$-dimensional Euclidean space, which
is $d$. It is worth noting that the corresponding heat kernel actually exhibits the
same decay rate as for Euclidean space. In order to see this it
suf\/f\/ices to note that
\[
C_{nm,kl,v}\leq\sqrt{{n+l\choose n-m+2v}{m+k\choose n-m+2v}}=
{m+k\choose n-m+2v} ,
\]
for $m+k=n+l$. Hence, for $t>0$,
\begin{gather*}
\left(e^{-t{\bf{\Delta}}}\right)_{nm,kl} \leq
\delta_{m+k,n+l}\sum_{v=0}^{\min\{m,l\}} {m+k\choose n-m+2v}
 \frac{t^{m+l-2v}}{(1+t)^{m+k+d}}\\
\hphantom{\left(e^{-t{\bf{\Delta}}}\right)_{nm,kl}}{}
= \delta_{m+k,n+l}\frac{t^{m+k}}{(1+t)^{m+k+d}}
\sum_{v=0}^{\min\{m,l\}} {m+k\choose n-m+2v}
 t^{-(n-m+2v)}\\
\hphantom{\left(e^{-t{\bf{\Delta}}}\right)_{nm,kl}}{}
 \leq \frac{t^{m+k}}{(1+t)^{m+k+d}}
\sum_{w=0}^{m+k}\,{m+k\choose w} t^{-w}\\
\hphantom{\left(e^{-t{\bf{\Delta}}}\right)_{nm,kl}}{}
= \frac{t^{m+k}}{(1+t)^{m+k+d}}\big(1+t^{-1}\big)^{m+k} = (1+t)^{-d} .
\end{gather*}
\end{remark}

\section{Existence of wave operators for small data}
\label{secth3}

\begin{proof}[Proof of Theorem~\ref{th3}.]
~~Both statements follow from~ \cite[Theorem 16]{Reed} once we verify
the following four conditions for $\alpha> 2d$ and some $\delta>0$,
where $p$ denotes the lowest degree of monomials occurring in $F$:
\begin{itemize}\itemsep=0pt
\item[i)] There exists a constant $c_1>0$ such that
\[
\Vert e^{-it{\bf{\Delta}}_\alpha}\phi\Vert_a\leq
c_1 |t|^{-d/2} \Vert\phi\Vert_{1,\alpha/2},\qquad |t|\geq 1 .
\]
\item[ii)] There exists a constant $c_2>0$ such that
\[
\Vert\phi\Vert_a\leq c_2 \Vert\phi\Vert_{2,\alpha} .
\]
\item[iii)] There  exists a constant $c_3>0$ such that
\[
\Vert\phi F(\phi^*\phi)-\psi F(\psi^*\psi)\Vert_{2,\alpha}
\leq c_3\big(\Vert\phi\Vert_a+\Vert\psi\Vert_a\big)^{2p-1}
\Vert\phi-\psi\Vert_{2,\alpha} ,\qquad \mbox{if}\ \ \Vert\phi\Vert_{2,\alpha},
\Vert\psi\Vert_{2,\alpha}\leq\delta .
\]
\item[iv)]  There  exists a constant $c_4>0$ such that
\begin{gather*}
\Vert\phi F(\phi^*\phi)-\psi F(\psi^*\psi)\Vert_{1,\alpha/2}\\
\qquad {} \leq c_4\Big\{\big(\Vert\phi\Vert_a+\Vert\psi\Vert_a\big)^{2p-2}
\Vert\phi-\psi\Vert_a
  +
\big(\Vert\phi\Vert_a+\Vert\psi\Vert_a\big)^{2p-1}
\Vert\phi-\psi\Vert_{2,\alpha}\Big\} ,\\
\qquad \qquad \mbox{if} \ \Vert\phi\Vert_{2,\alpha},
\Vert\psi\Vert_{2,\alpha}\leq\delta .
\end{gather*}
\end{itemize}
Moreover, it is required that the constants $c_3$, $c_4$ can be chosen
arbitrarily small by choosing $\delta$ small enough.

To be specif\/ic, the stated basic assumptions about $F$ ensure that
 $2p-1\geq 1$ and $\frac d2(2p-1)> 1$ for all
$d\geq 1$, which implies \eqref{solsmall1} by  \cite[Theorem
16]{Reed}. If, in addition, $F$ has no linear term, we have $2p-1>1$
and \eqref{solsmall2} follows similarly.

Condition i) above is a particular case of Theorem~\ref{th2}. That
condition ii) holds, follows from
\begin{gather*}
\Vert\phi\Vert_a^2 \leq \Vert\widetilde\phi\Vert^2_2 = \Vert\phi\Vert_{2,0}^2 \leq
\sum_{m,n}b_{mn}^{2\alpha}\big|\phi_{mn}\big|^2 = \Vert\phi\Vert^2_{2,\alpha},\qquad \alpha\geq0,
\end{gather*}
where we have used that the Hilbert--Schmidt norm dominates the operator
norm in the f\/irst step and that $b^{2\alpha}_{mn}\geq 1$ for
$\alpha\geq 0$ in the third step.

To establish iii) and iv) we make use of Lemma~\ref{product}. As a
consequence of the inequality
$\Vert\phi\,\psi\Vert_a\leq\Vert\phi\Vert_a \Vert\psi\Vert_a$, which
follows immediately from
\[
(\widetilde{\phi\psi})_{mn} = \bigg|\sum_k\phi_{mk} \psi_{kn}\bigg| \leq \sum_k|\phi_{mk}|\,|\psi_{kn}| = (\widetilde\phi \widetilde\psi)_{mn} ,
\]
the f\/irst part of the Lemma implies
\begin{gather}
\label{prodnorm1}
\Vert\phi_1\cdots\phi_r\Vert_{2,\alpha}\leq
C(\alpha,r)  \sum_{i=1}^r\Vert\phi_i\Vert_{2,\alpha} \prod_{j\ne i}\Vert\phi_j\Vert_a ,
\end{gather}
valid for $\alpha\geq 0$ and some constant $C(\alpha,r)$ depending
on $\alpha$ and $r$ only.

The second statement of Lemma~\ref{product} implies, for $\alpha>2d$,
\begin{gather*}
\Vert\phi_1\cdots\phi_r\Vert_{1,\alpha/2} \leq
C_{2,\alpha/2} \bigg(\Vert\phi_1\Vert_{2,\alpha} \bigg\Vert\prod_{i=2}^r\phi_i\bigg\Vert_2 +
 \Vert\phi_1\Vert_2 \bigg\Vert\prod_{i=2}^r\phi_i\bigg\Vert_{2,\alpha}\bigg)\\
 \hphantom{\Vert\phi_1\cdots\phi_r\Vert_{1,\alpha/2}}{}
 \leq
2C_{2,\alpha/2} \Vert\phi_1\Vert_{2,\alpha} \bigg\Vert\prod_{i=2}^r\phi_i\bigg\Vert_{2,\alpha}
\end{gather*}
 where, in the second step, it has been used that $\Vert\phi\Vert_{2,\alpha}$ is an
increasing function of $\alpha$.

The last expression can now be estimated by making use of~\eqref{prodnorm1}. Thus we obtain
\begin{gather}\label{prodnorm2}
\Vert\phi_1\cdots\phi_r\Vert_{1,\alpha/2}
\leq 2 C(\alpha,r-1)C_{2,\alpha/2} \Vert\phi_1\Vert_{2,\alpha} \sum_{i=2}^r\Vert\phi_i\Vert_{2,\alpha}
\prod_{j\ne 1,i}\Vert\phi_j\Vert_a .
\end{gather}
Now, let $p$ be an arbitrary positive integer and write
\begin{gather*}
\phi |\phi|^{2p}-\psi |\psi|^{2p}
=\sum_{i=0}^{p-1}\phi |\phi|^{2(p-1-i)}
(\phi^* (\phi-\psi)+(\phi^*-\psi^*) \psi) |\psi|^{2i}
+(\phi-\psi) |\psi|^{2p} .
\end{gather*}
Since the norms $\Vert\cdot\Vert_{2,\alpha}$, $\Vert\cdot\Vert_a$, $\Vert\cdot\Vert_{1,\alpha}$
are  $*$-invariant, an application of~\eqref{prodnorm1}
and~\eqref{prodnorm2} yields the inequalities
\begin{gather*}
\Vert\phi |\phi|^{2p}-\psi |\psi|^{2p}\Vert_{2,\alpha}
\leq c'_3 \big(\Vert\phi\Vert_{2,\alpha}+\Vert\psi\Vert_{2,\alpha}\big)
\big(\Vert\phi\Vert_a+\Vert\psi\Vert_a\big)^{2p-1}\Vert\phi-\psi\Vert_{2,\alpha} ,
\end{gather*}
and
\begin{gather*}
 \Vert\phi |\phi|^{2p}-\psi |\psi|^{2p}\Vert_{1,\alpha/2}
 \leq c'_4
\Big\{\big(\Vert\phi\Vert_{2,\alpha}+\Vert\psi\Vert_{2,\alpha}\big)^2
\big(\Vert\phi\Vert_a+\Vert\psi\Vert_a\big)^{2p-2}\Vert\phi-\psi\Vert_a\\
\phantom{\Vert\phi |\phi|^{2p}-\psi |\psi|^{2p}\Vert_{1,\alpha/2}\leq}{}
+
\big(\Vert\phi\Vert_{2,\alpha}+\Vert\psi\Vert_{2,\alpha}\big)
\big(\Vert\phi\Vert_a+\Vert\psi\Vert_a\big)^{2p-1}\Vert\phi-\psi\Vert_{2,\alpha}\Big\} ,
\end{gather*}
respectively, for suitable constants $c'_3$, $c'_4$, where the
obvious inequality
\[
\Vert\phi\Vert_a \leq \Vert\phi\Vert_{2,\alpha} ,\qquad \alpha\geq 0 ,
\]
has also been used. This evidently proves iii) and iv), with $c_3$ and
$c_4$ of order $\delta$, if $F$ is a~monomial of degree $p$. The same bounds then follow immediately for an
arbitrary polynomial~$F$ without constant term, if~$p$ denotes the
lowest degree of monomials occurring in~$F$.
\end{proof}

\begin{remark}\label{remscat}
As pointed out in \cite[Theorem 17]{Reed}, the
assumptions of Theorem~\ref{th3} also ensure the existence of a
scattering operator def\/ined for suf\/f\/iciently small initial data. More
precisely, if $\delta_0>0$ is small enough, there exists, for each
$\phi_-\in\Sigma_\alpha$ with $\Vert\phi_-\Vert_{2,\alpha}\leq \delta_0$, a unique
$\phi_+\in\Sigma_\alpha$ with $\Vert\phi_-\Vert_{2,\alpha}\leq 2\delta_0$, such
that the solution $\phi^-(t)$ to \eqref{NLSch4-} fulf\/ills
\[
\Vert\phi^-(t)- e^{-it{{\bf{\Delta}}_\alpha}}\phi_+\Vert_{2,\alpha}  \to
0,\qquad\mbox{for}\ \ t\to +\infty ,
\]
and the so def\/ined mapping
$S: \{\phi_-\in\Sigma_\alpha\,|\;\Vert\phi_-\Vert_{2,\alpha}\leq
\delta_0\}\to\{\phi_+\in\Sigma_\alpha\,|\;\Vert\phi_+\Vert_{2,\alpha}\leq 2\delta_0\}$,
called the {\it scattering operator},
is injective and continuous w.r.t.\ the topology def\/ined by
$\Vert\cdot\Vert_{2,\alpha}$.
\end{remark}

\section{The diagonal case}
\label{secdiag}

In this f\/inal section we discuss brief\/ly equation~\eqref{NLSch} when
restricted to functions $\varphi(x,t)$ that correspond under the Weyl
map to operators that are diagonal w.r.t.\ the basis
$\{|n\rangle\}$. This means that the operators $\phi(t)$ commute with
the number operators $a_k^*a_k$, $k=1,\dots,d$. Since  $a_k^*a_k$ corresponds
under the Weyl map to the generator of rotations in the
$(x_k,x_{k+d})$-plane, the functions $\varphi(x,t)$ in question are
invariant under such rotations, i.e.\ under the action of the
$d$-fold product of $SO(2)$. Since, clearly, $\phi F(|\phi|^2)$ is
diagonal and Hilbert--Schmidt if $\phi$ is and ${\bf\Delta}$ naturally
restricts to a self-adjoint operator on the Hilbert subspace
\[
\mathcal H_{\rm diag} = \{\phi\in{\cal H}\,|\; \phi\;\, \mbox{diagonal}\},
\]
(see \cite{DJN}) it follows that the equations \eqref{NLSch2} and
\eqref{NLSch3} make sense as equations on $\mathcal H_{\rm diag}$.

We f\/irst note that {\it Theorem~\ref{th1} still holds with ${\cal H}$
replaced by $\mathcal H_{\rm diag}$}. In fact, the proof of a) applies without additional
changes and part b) follows likewise since the operator $\phi_0$ in
the proof obtained from \cite{DJN} is diagonal.

Concerning Theorem~\ref{th2} we have the stronger decay estimate
\begin{gather}\label{decay3}
\big|\big(e^{it{\bf\Delta}}\big)_{nn,mm}\big| \leq  C |t|^{-d}(1+\ln |t|)^d ,\qquad |t|\geq 1 ,
\end{gather}
from Lemma~\ref{logestimate} above as noted previously.
Note also that for $\phi\in\mathcal H_{\rm diag}$ the norms $\Vert\phi\Vert_{p,\alpha}$
are independent of $\alpha$ and the identities
\begin{align*}
\Vert\phi\Vert_a\;=\Vert\widetilde\phi\Vert_{\rm op}=\sup_n|\phi_{nn}|=\Vert\phi\Vert_{\rm op}\,,\qquad
\Vert\phi\Vert_{1,\alpha}=\sum_n|\phi_{nn}| =\Vert\phi\Vert_1
\end{align*}
hold. In view of \eqref{decay3} the decay estimate in Theorem~\ref{th2} is
hence replaced by
\begin{gather}\label{decaydiag}
\Vert e^{-it{\bf\Delta}}\phi\Vert_{\rm op} \leq
C |t|^{-d} (1+\log|t|)^d \Vert\phi\Vert_1 ,\qquad |t|\geq 1 .
\end{gather}
We therefore def\/ine the subspace of diagonal scattering states by
\begin{gather*}
\Sigma_{\rm diag}:=\left\{\phi\in\mathcal H_{\rm diag}\, |\; |||\phi|||<\infty\right\} ,
\end{gather*}
where
\begin{gather*}
|||\phi|||:=
\Vert\phi\Vert_2 +\sup_{|t|\geq 1}  |t|^d (1+\log |t|)^{-d}
\Vert e^{-it{\bf\Delta}}\phi\Vert_{\rm op} .
\end{gather*}

Using the well known inequalities
\begin{gather*}
\Vert\phi \psi\Vert_2 \leq \frac12\big(\Vert\phi\Vert_{\rm op} \Vert\psi\Vert_2 +
\Vert\phi\Vert_2\,\Vert\psi\Vert_{\rm op}\big)\qquad\mbox{and}\qquad
\Vert\phi \psi\Vert_1\leq\Vert\phi\Vert_2\,\Vert\psi\Vert_2
\end{gather*}
as a replacement for Lemma~\ref{product}, we obtain the estimates
\begin{gather*}
\Vert\phi |\phi|^{2p}-\psi |\psi|^{2p}\Vert_2
 \leq  C_1 \big(\Vert\phi\Vert_2+\Vert\psi\Vert_2\big)
\big(\Vert\phi\Vert_{\rm op}+\Vert\psi\Vert_{\rm op}\big)^{2p-1}\Vert\phi-\psi\Vert_2 ,\\
\Vert\phi |\phi|^{2p}-\psi |\psi|^{2p}\Vert_{1}
\leq C_2
\Big\{\big(\Vert\phi\Vert_2+\Vert\psi\Vert_2\big)^2
\big(\Vert\phi\Vert_{\rm op}+\Vert\psi\Vert_{\rm op}\big)^{2p-2}\Vert\phi-\psi\Vert_{\rm op}\\
\hphantom{\Vert\phi |\phi|^{2p}-\psi |\psi|^{2p}\Vert_{1} \leq}{}
+
\big(\Vert\phi\Vert_2+\Vert\psi\Vert_2\big)
\big(\Vert\phi\Vert_{\rm op}+\Vert\psi\Vert_{\rm op}\big)^{2p-1}\Vert\phi-\psi\Vert_2\Big\} ,
\end{gather*}
by arguments analogous to those in the proof of
Theorem~\ref{th3}. It hence follows that the conclusions of
Theorem~\ref{th3} hold in the diagonal case with $\Vert\cdot\Vert_2$ replacing
$\Vert\cdot\Vert_{2,\alpha}$ and $|||\cdot|||$ replacing
$|||\cdot|||_\alpha$, provided $F$ has no constant term and, in addition, no
linear term if $d=1$. According to the diagonal version of
Theorem~\ref{th2}, we have in this
case also global existence of solutions to the Cauchy problem on
$\mathcal H$. It then follows by standard arguments \cite[Theorem 19]{Reed}
that the restriction to small data at $\pm\infty$ can be dropped,
i.e.\ the wave operators $\Omega_\pm$ can be def\/ined on the full space
$\Sigma_{\rm diag}$ if $F$ has no linear term. More precisely we have:

\begin{theorem}\label{th4} Assume the polynomial $F$ is real and contains
no constant and linear terms. Then, for all $\phi_\pm\in\Sigma_{\rm diag}$,
equations \eqref{NLSch4+} and \eqref{NLSch4-} have unique globally
defined continuous solutions $\phi^\pm:\mathbb R\to \Sigma_{\rm diag}$ fulfilling
\begin{gather}
\Vert\phi^\pm(t)-e^{-it{\bf\Delta}}\phi_\pm\Vert_{2}  \to
0\qquad\mbox{for}\quad t\to\pm\infty ,\label{decay1}\\
 |||e^{it{\bf\Delta}}\phi^\pm(t)-\phi_\pm |||  \to
0\qquad\mbox{for}\quad t\to\pm\infty .\nonumber
\end{gather}
\end{theorem}

\begin{remark}
This result allows us to def\/ine wave operators $\Omega_\pm:
\Sigma_{\rm diag}\to\Sigma_{\rm diag}$ by equation~\eqref{waveop} that are
injective and uniformly continuous on balls in $\Sigma_{\rm diag}$.
In general, $\Omega_\pm$ are not surjective: choosing $F$ to
satisfy the assumptions of Theorem~\ref{th1}$b)$ the oscillating
solution $\phi(t)=e^{i\omega t}\phi_0$ to the Cauchy problem (with
$t_0=0$) fulf\/ills $\Vert\phi(t)\Vert_1 = \Vert\phi_0\Vert_1 < \infty$
by \cite[Lemma 2]{DJN}. Hence, $\phi_0\in \Sigma_{\rm diag}$ by
\eqref{decaydiag}. On the other hand, since the unique solution
$\phi(t)$ to the Cauchy problem with initial value $\phi_0$ at $t=0$
has constant operator norm, it cannot fulf\/ill \eqref{decay1} for any
$\phi_\pm\in\Sigma_{\rm diag}$. Even in this case, an appropriate
characterization of the images of
$\Omega_\pm$ remains an open question.

\end{remark}

\pdfbookmark[1]{References}{ref}
\LastPageEnding


\begin{thebibliography}{99}

\footnotesize\itemsep=0pt

\bibitem{CaLi}
Cazenave  T., Lions P.L.,
Orbital stability of standing waves for some nonlinear Schr\"odinger equations,
\href{http://dx.doi.org/10.1007/BF01403504}{{\it Comm. Math. Phys.}} {\bf 85} (1982), 549--561.

\bibitem{Caz}
Cazenave T.,
Semilinear Schr{\"o}dinger equations,
{\it Courant Lecture Notes in Mathematics}, Vol.~10, Amer. Math. Soc., Providence, RI, 2003.

\bibitem{DJ}
Durhuus  B., Jonsson T.,
Noncommutative waves have inf\/inite propagation speed,
\href{http://dx.doi.org/10.1088/1126-6708/2004/10/050}{{\it J. High Energy Phys.}} {\bf 2004} (2004), no.~10, 050, 13~pages,
\href{http://arxiv.org/abs/hep-th/0408190}{hep-th/0408190}.

\bibitem{DJN}
Durhuus B., Jonsson T., Nest R.,
The existence and stability of noncommutative scalar solitons,
\href{http://dx.doi.org/10.1007/s00220-002-0721-4}{{\it Comm. Math. Phys.}} {\bf 233} (2003), 49--78,
\href{http://arxiv.org/abs/hep-th/0107121}{hep-th/0107121}.

\bibitem{erdel}
Erd{\'e}lyi T., Magnus  A.P., Nevai P.,
Generalized Jacobi weights, Christof\/fel functions, and Jacobi polynomials,
\href{http://dx.doi.org/10.1137/S0036141092236863}{{\it SIAM J. Math. Anal.}} {\bf 25} (1994), 602--614.

\bibitem{GGBISV}
Gayral V., Gracia-Bond\'{\i}a J.M., Iochum B., Sch\"ucker T.,  V\'arilly J.C.,
Moyal planes are spectral triples,
\href{http://dx.doi.org/10.1007/s00220-004-1057-z}{{\it Comm. Math. Phys.}} \textbf{246} (2004), 569--623,
\href{http://arxiv.org/abs/hep-th/0307241}{hep-th/0307241}.

\bibitem{Ginibre}
Ginibre J.,
An introduction to nonlinear Schr{\"o}dinger equations,
in Nonlinear Waves (Sappore 1995), Editors R.~Agemi, Y.~Giga and T.~Ozawa,
{\it GAKUTO Internat. Ser. Math. Sci. Appl.}, Gakk{\=o}tosho, Tokyo, 1997,  85--133.

\bibitem{GiVe}
Ginibre J., Velo G.,
Time decay of f\/inite energy solutions of the nonlinear Klein--Gordon and Schr{\"o}dinger equations,
{\it Ann. Inst. H. Poincar{\'e} Phys. Th{\'e}or.} {\bf 43} (1985), 399--442.

\bibitem{GMS}
Gopakumar R., Minwalla S., Strominger A.,
Noncommutative solitons,
\href{http://dx.doi.org/10.1088/1126-6708/2000/05/020}{{\it J. High Energy Phys.}} {\bf 2000}, (2000), no.~5, 020, 27~pages,
\href{http://arxiv.org/abs/hep-th/0003160}{hep-th/0003160}.

\bibitem{GSS}
Grillakis M., Shatah J., Strauss W.,
Stability theory of solitary waves in the presence of symmetry.~I,
\href{http://dx.doi.org/10.1016/0022-1236(87)90044-9}{{\it J.~Funct. Anal.}} {\bf 74} (1987), 160--197.

\bibitem{HKL}
Harvey J., Kraus P., Larsen F.,
Exact noncommutative solitons,
 \href{http://dx.doi.org/10.1088/1126-6708/2000/12/024}{{\it J. High Energy Phys.}} {\bf 2000}, (2000), no.~12, 024, 24~pages,
 \href{http://arxiv.org/abs/hep-th/0010060}{hep-th/0010060}.

\bibitem{Krasikov}
Krasikov I.,
An upper bound on Jacobi polynomials,
\href{http://dx.doi.org/10.1016/j.jat.2007.04.008}{{\it J. Approx. Theory}} \textbf{149} (2007), 116--130,
\href{http://arxiv.org/abs/math.CA/0610111}{math.CA/0610111}.

\bibitem{Reed}
Reed M.,
Abstract non-linear wave equations,
{\it Lecture Notes in Mathematics}, Vol.~507, Springer-Verlag,
  Berlin~-- New York, 1976.

\bibitem{RS}
Reed  M., Simon B.,
Methods of modern mathematical physics, Vol.~II, Academic Press, New York~-- London, 1975.

\bibitem{ShStrauss}
Shatah J., Strauss W.A.,
Instability of nonlinear bound states,
\href{http://dx.doi.org/10.1007/BF01212446}{{\it Comm. Math. Phys.}} \textbf{100} (1985), 173--190.


\bibitem{Szego}
Szeg\"o G.,
Orthogonal polynomials,  {\it American Mathematical Society Colloquium Publications}, Vol.~23, Amer. Math. Soc., Providence, RI, 1959.

\bibitem{Strauss} 
Strauss W.A.,
Nonlinear wave equations,
{\it CBMS Regional Conference Series in Mathematics}, Vol.~73, Amer. Math. Soc., Providence, RI, 1989.

\bibitem{Raimar} Wulkenhaar R., Renormalisation of noncommutative $\varphi_4^4$-theory to all orders, Habilitation Thesis, University of Vienna, 2003, available at \url{http://www.math.uni-muenster.de/u/raimar/}.

\end{thebibliography}
\end{document}